\documentclass[a4paper,10pt]{llncs}
\usepackage[utf8]{inputenc}

\usepackage[draft]{todonotes}

\usepackage{amsmath, amsfonts, amssymb}
\usepackage{mathabx}
\usepackage{longtable}

\usepackage{hyperref}
\usepackage{url}
\usepackage{bussproofs}

\usepackage{alltt} 

\usepackage{listings}
\usepackage{pgf,tikz}
\usetikzlibrary{arrows,automata}

\newcommand{\ie}{{\it i.e.\ }}
\newcommand{\eg}{{\it e.g.\ }}
\newcommand{\dL}{d\mathcal{L}}
\newcommand{\disc}{\textbf{disc}}
\newcommand{\cont}{\textbf{cont}}
\newcommand{\ctrl}{\textit{ctrl}}
\newcommand{\plant}{\textit{plant}}
\newcommand{\com}{\text{cmp}}
\newcommand{\ODEand}{\ \mathbin{\&}\ }
\newcommand{\KeYmaeraX}{KeYmaera~X}

\newcommand{\et}{\wedge} 
\newcommand{\ou}{\vee} 
\newcommand{\imply}{\rightarrow}
\newcommand{\Env}{\mathcal{E}} 
\newcommand{\Init}{Init} 
\newcommand{\para}{\ \otimes \ } 
\newcommand{\eval}[1]{\ldbrack #1 \rdbrack}
\newcommand{\fout}{\mathit{fout}}
\newcommand{\fin}{\mathit{fin}}
\newcommand{\C}{{\mathcal{C}}}
\newcommand{\reactivity}{\ensuremath{\delta}}
\newcommand{\controllability}{\ensuremath{\Delta}}
\newcommand{\timestamp}{\ensuremath{\tau}}
\newcommand{\wlctrl}{\textit{wlctrl}}
\newcommand{\wl}{\textit{wl}}

\newcommand*{\ptest}[1]{\ensuremath{?#1}}
\newcommand*{\prepeat}[2][*]{\ensuremath{{#2}^{#1}}}
\newcommand*{\pumod}[2]
	{#1\hspace{-0.09em}\mathrel{{:}{=}}\hspace{-0.09em}#2}

\newcommand*{\pchoice}[2]{\ensuremath{{#1}\cup{#2}}}
\newcommand*{\pevolvein}[2]{\ensuremath{\{#1 \ODEand #2\}}}
\newcommand*{\D}[1]{#1'}
\newcommand*{\limply}{\ensuremath{\rightarrow}}
\newcommand*{\dbox}[2]{\ensuremath{[#1]#2}}

\newcommand*{\freevars}[1]{\ensuremath{\text{FV}(#1)}}
\newcommand*{\boundvars}[1]{\ensuremath{\text{BV}(#1)}}
\newcommand*{\vars}[1]{\ensuremath{\text{V}(#1)}}

\title{Parallel Composition and Modular Verification of Computer Controlled Systems in Differential Dynamic Logic\thanks{This material is based upon work supported by the United States Air Force and DARPA under Contract No. FA8750-18-C-0092. Any opinions, findings and conclusions or recommendations expressed in this material are those of the author(s) and do not necessarily reflect the views of the United States Air Force and DARPA.}}
\author{Simon Lunel\inst{1,2} \and Stefan Mitsch\inst{3} \and Benoit Boyer\inst{1} \and Jean-Pierre Talpin\inst{2}}
\institute{Mitsubishi Electric R\&D Centre Europe, 1 allée de Beaulieu, CS 10806, 35708 Rennes CEDEX~7, FRANCE \email{b.boyer@fr.merce.mee.com}
\and Inria, Centre de recherche  Rennes - Bretagne - Atlantique, Campus universitaire de Beaulieu, 35042 Rennes Cedex, FRANCE \email{jean-pierre.talpin@inria.fr}
\and Computer Science Department, Carnegie Mellon University, Pittsburgh, PA, USA \email{smitsch@cs.cmu.edu}}

\begin{document}

\tikzstyle{myarrow}=[->, >=triangle 45, thick]
\tikzstyle{every state}=[shape=rectangle, rounded corners, align=center, fill=white, draw=black, text=black]

\maketitle

\begin{abstract}
Computer-Controlled Systems (CCS) are a subclass of hybrid systems where the periodic relation of control components to time is paramount. Since they additionally are at the heart of many safety-critical devices, it is of primary importance to  correctly model such systems and to ensure they function correctly according to safety requirements. Differential dynamic logic $\dL$ is a powerful logic to model hybrid systems and to prove their correctness.  We contribute a component-based modeling and reasoning framework to $\dL$ that separates models into components with timing guarantees, such as reactivity of controllers and controllability of continuous dynamics. Components operate in parallel, with coarse-grained interleaving, periodic execution and communication. We present techniques to automate system safety proofs from isolated, modular, and possibly mechanized proofs of component properties parameterized with timing characteristics.
\end{abstract}

\section{Introduction}
\label{sec:Introduction}

\emph{A computer-controlled system} (CCS) is a hybrid system with discrete hardware-software components that control a specific physical phenomenon, \eg the water level of a tank in a water-recycling plant. CCSs are widely used in industry to monitor time-critical and safety-critical processes. While CCS defines a large class of hybrid systems, systems mixing physical phenomena and natural discrete interactions ({\it e.g.} a bouncing-ball) are neither CCSs nor the focus of this work, although most could easily be given verification models in $\dL$.
Tools to model, verify, and design CCSs need to capture mixed discrete and continuous dynamics, as well as mixed logical, discretized real-time and continuous time in, resp., computer programs, electronics, and physics models.

CCSs are difficult to model since they subsume the problem of designing a software controller and its real-timed hardware. Our aim is to develop a \emph{component-based approach} to engineer such systems in a modular manner while accounting for time domain boundaries across components. In component-based design, a system is constructed from smaller elements that are modeled and individually verified, then assembled and checked for consistency to form larger components and subsystems. The CCS components typically execute in parallel, and concurrency must be accounted for in the envisioned verification framework. 

In this paper, we contribute a component-based verification technique that aims at the definition of a bottom-up and modular verification methodology through a \emph{correct-by-construction} system design methodology, in which component \emph{contracts} formalize the domain, timing, and invariants of components. The proof of a system model is built by assembling the contracts of its components through formally defined composition mechanisms. Contracts-based approaches have been successfully implemented for several paradigms such as programming languages~\cite{frama-C} or automata~\cite{Benveniste2012}, because contracts are very efficient to make proofs easier scalable. Following the component-based design widely used for CCS, contracts provide a natural way to get modularity and abstraction in proofs.

To meet the specific time-criticality requirements of CCSs, we start from earlier compositionality results in $\dL$~\cite{Lunel2017} to elaborate a timed model of parallel composition that forms the foundation of our modeling and verification framework.
In Sec.~\ref{sec:CCS}, we detail the modeling part and the verification part of our framework on a simple system where only one reactive controller monitors a plant. In Sec.~\ref{sec:TimedParallelComposition}, we generalize it step-by-step to systems where multiple controllers monitor multiple parallel plants; we show how to compose: multiple reactive controllers into a component called Multi-Choice Reactive Controllers ($\textbf{MRCtrl}$, see Sec.~\ref{subsec:CompoDisc}); multiple controllable plants together (see Sec.~\ref{subsec:CompoCont}); $\textbf{MRCtrl}$ with controllable plants to form Multi Computer-Controlled Systems ($\textbf{MCCS}$, see Sec.~\ref{subsec:CompoDiscCont}); and finally how $\textbf{MCCS}$ compose (see Sec.~\ref{subsec:CompoGeneral}).

\section{Differential Dynamic Logic}
\label{sec:ddl}

\label{sec:diff-dyn-logic}

This section briefly recalls \emph{differential dynamic logic} ($\dL$ \cite{DBLP:journals/jar/Platzer17}) and its proof system, which is implemented in the theorem prover \KeYmaeraX~\cite{DBLP:conf/cade/FultonMQVP2015}.

In $\dL$, hybrid programs are used as a programming language for expressing the combined discrete and continuous dynamics of hybrid systems (the programs operate over mathematical reals).
The syntax and semantics of hybrid programs is summarized in Tab.~\ref{tab:hybrid-programs}. The set of reachable states from state $\nu$ by hybrid program $\alpha$ is noted $\rho_\nu(\alpha)$, and $\ldbrack x \rdbrack_\nu$ denotes the value of $x$ at state $\nu$.
\begin{table}[htbp]
\centering
\caption{Syntax and semantics of hybrid programs}
\label{tab:hybrid-programs}
\begin{tabular}{ll}
Program & Semantics   \\\hline
$\ptest{\phi}$         & Test whether formula $\phi$ is true, abort if false:
$\rho_\nu (\ptest{\varphi}) = \{ \nu \mid \nu \models \varphi \}$\\
$\pumod{x}{\theta}$    & Assign value of term $\theta$ to variable $x$:\\
& $\rho_\nu (\pumod{x}{\theta}) = \{\omega \mid \omega = \nu \mbox{ except that } \ldbrack x \rdbrack_\omega = \ldbrack \theta \rdbrack_\nu \}$\\
$\pevolvein{\D{x}=\theta}{H}$ & Evolve ODE $\D{x} = \theta$ for any time $t{\geq}0$\\
& with evolution domain constraint $H$ true throughout:\\
& $\rho_\nu (\pevolvein{\D{x} = \theta}{H}) = \{ f(r) \mid f(0) = \nu \mbox{ and for any duration } r \geq 0$\\ 
& \hfill $f(r) \models \D{x} = \theta \mbox{ and } f(r) \models H\}$\\
$\alpha;\beta$         & Run $\alpha$ followed by $\beta$ on resulting state(s): $\rho_\nu (\alpha; \beta) = \bigcup_ {\omega \in \rho_\nu (\alpha)} \rho_\omega (\beta)$\\
$\pchoice{\alpha}{\beta}$      & Run either $\alpha$ or $\beta$ non-deterministically: $\rho_\nu (\pchoice{\alpha}{\beta}) = \rho_\nu (\alpha) \cup \rho_\nu (\beta)$\\
$\prepeat{\alpha}$ & Repeat $\alpha$ $n$ times, for any $n\in\mathbb{N}$:\\ 
& $\rho(\prepeat{\alpha}) = \bigcup_{n \in \mathbb{N}} \rho (\alpha^n)\,\mbox{with}\, \alpha^0 = \ptest{\top} \, \mbox{ and }\,\alpha^{n + 1} = \alpha^n ; \alpha$
\end{tabular}
\end{table}
Hybrid programs include discrete assignment $\pumod{x}{\theta}$ and tests $\ptest{\phi}$, as well as combinators for non-deterministic choice ($\alpha\cup\beta$), sequential composition ($\alpha ; \beta$), and non-deterministic repetition ($\prepeat{\alpha}$).
The notation $\pevolvein{\D{x} = \theta}{H}$ denotes an ordinary differential equation (ODE) system (derivatives with respect to time) of the form $\D{x}_1 = \theta_1$, \dots, $\D{x}_n = \theta_n$ within evolution domain $H$. 
For example, the ODE $\pevolvein{\D{t} = 1}{t \geq 0}$ describes that variable $t$ evolves with constant slope $1$, where $t \geq 0$ discards any negative values.
Formulas of $\dL$ formalize properties:
\begin{definition}[$\dL$ formulas]
The formulas $\phi,\psi$ of $\dL$ relevant in this paper consist of the following operators:
\begin{equation*}
\phi,\psi ::= \phi\land\psi\ |\ \phi\lor\psi\ |\ \phi\limply\psi\ |\ \neg\phi\ |\ \theta_1 \sim \theta_2
      |\ \forall{x}{~\phi}\ |\ \exists{x}{~\phi}\ |\ \dbox{\alpha}{\phi}\
\end{equation*}
\end{definition}
Connectives $\phi\land\psi$, $\phi\lor\psi$, $\phi\limply\psi$,$\neg\phi$, $\forall{x}{~\phi}$, and $\exists{x}{~\phi}$ are according to classical first-order logic.
Formula $\theta_1 \sim \theta_2$ are any comparison operator $\sim\mathop{\in}\{\leq,<,=,\neq,>,\geq\}$ and $\theta_i$ are real-valued terms in operators $\{+,-,\cdot,/\}$.
The modal operator $\dbox{\alpha}{\phi}$ is true iff $\phi$ holds in all states reachable by program $\alpha$. 


\begin{table}
\centering
\caption{Semantic of formulas}
\label{tab:semantic-formulas}

 \[
 \begin{array}{l c l}
   \nu \models \theta_1 \sim \theta_2 &\Leftrightarrow& \ldbrack \theta_1 \rdbrack_\nu \sim \ldbrack \theta_2 \rdbrack_\nu \mbox{, } \sim\;\in \{<, \leq, =, \geq, > \} \\
   \nu \models \varphi \rightarrow \psi &\Leftrightarrow& \nu \models \varphi \Rightarrow \nu \models \psi \\
   \nu \models \forall x\; \varphi & \Leftrightarrow& \omega \models \varphi, \forall\omega|\ \forall z, z \neq x \Rightarrow \eval{z}_\omega = \eval{z}_\nu\\
   \nu \models [\alpha] \varphi &\Leftrightarrow& \forall \omega \in \rho_\nu(\alpha),\  \omega \models \varphi\\
   \nu \not\vDash \bot \\
  \end{array}
 \]
\end{table}

The notion of \emph{free and bound variables} is defined in the static semantics of $\dL$~\cite{DBLP:journals/jar/Platzer17} and useful to characterize the interaction of a program $\alpha$ with its context. 
It is computed from the syntactic structure of programs: the bound variables $\boundvars{\alpha}$ can be updated by assignments (\eg $\pumod{x}{10}$)
or ODEs (\eg $\D{x} = 3$) in $\alpha$, whereas free variables $\freevars{\alpha}$ are those that the program depends on.
For example, in program $\alpha \equiv (\pchoice{\pumod{v}{a}}{\pumod{v}{2}});\pevolvein{\D{x}=v}{x \leq 5}$ the free variables are $\freevars{\alpha} = \{a,x\}$ (the variable $v$ is not free, because it is bound on all paths of $\alpha$ and so the result of $\alpha$ does not depend on the initial value of $v$; even though also modified, variable $x$ is free because the result of $\alpha$ depends on the initial value of $x$) and the bound variables are $\boundvars{\alpha} = \{v,x\}$ (variable $a$ is not bound because it is not modified anywhere in the program).
We use $\vars{\alpha}$ to denote $\boundvars{\alpha} \cup \freevars{\alpha}$.

In~\cite{Lunel2017}, a component model $C_i =\prepeat{(\pchoice{\disc_i}{\cont_i})}$
 is evaluated as the non-deterministic interleaving of its discrete specifications $\disc_i$ and the ODE $\cont_i = \pevolvein{\D{x_i}=\theta_i}{H_i}$. For $i\in\{1,2\}$, the parallel composition of $C_1$ and $C_2$ builds a component of the same structure, {\it i.e.}
the dicrete parts $\disc_1$ and $\disc_2$ are non-deterministically interleaved within the evolution of the ODE obtained from the mathematical composition of $\cont_1$ and $\cont_2$.  In $\dL$, it is defined
\begin{equation}\label{eq:choiceparallel}
  C_1\otimes C_2=\prepeat{(\pchoice{\pchoice{\disc_1}{\disc_2}}{{\pevolvein{\D{x_1}=\theta_1, \D{x_2}=\theta_2}{H_1 \et H_2}}})}
\end{equation}

The logic $\dL$ further  enjoys a proof calculus~\cite{DBLP:conf/lics/Platzer12b,DBLP:journals/jar/Platzer17,Platzer18} based on uniform substitution from axioms. 
Its base axioms are those of a classical first-order sequent calculus.  
A second set of axioms is dedicated to goals of the form $[\alpha]\phi$, supporting syntactical deconstruction of hybrid programs $\alpha$. 
The last set of axioms is dedicated to iteration and ODEs~\cite{DBLP:conf/lics/PlatzerT18}.

\section{Computer-Controlled Systems}
\label{sec:CCS}

We present a component-based approach to model and verify Computer Controlled Systems (CCS) based on the parallel composition pattern proposed in~\cite{Lunel2017}. 
We construct the proof of a CCS from the isolated sub-proofs of its components by \emph{syntactically} decomposing the CCS using the axioms of $\dL$, so that the theorems presented here can be implemented as tactics in the theorem prover \KeYmaeraX.
This enables automation to reduce the proof complexity of analyzing a CCS to that of modularly analyzing its components.

In this section, we introduce the necessary concepts to adapt the framework of~\cite{Lunel2017} to systematically model CCSs modularly. 
We achieve modularity by limiting the ways in which the free and bound variables of different components may overlap, and by taking into account the timing constraints of CCSs.
The idea is to analyze the \emph{controllability} of the plant, \ie the period of time it can evolve safely without intervention of a controller, and the \emph{reactivity} of the controller, \ie the execution period of the controller.  These concepts satisfy the associativity property of parallel composition and the ability to retain contracts from~\cite{Lunel2017}. 

\subsection{Modeling CCS}
\label{subsec:ModelingCCS}


A CCS is classically composed of a \emph{controller} and a \emph{plant}. 
The controller measures the state of the plant through sensing and regulates the behavior of the plant through actuation. 
For example, the controller in the water tank regulates the water level by opening or closing a faucet.
%
%
%
The key trait of a CCS is the periodic execution of the controller to regulate the plant. 
We associate periodic values $\reactivity$ and $\controllability$ with the controller and the plant, respectively: 
Control \emph{reactivity} $\reactivity$ models the period in which control is guaranteed to happen. 
Plant \emph{controllability} $\controllability$ models how long a plant can evolve safely without control intervention.

\subsubsection{Time.} To make timed reasoning available to any component, we use the ODE $\textbf{Time}(t)\doteq\pevolvein{\D{t} = 1}{t \geq 0}$. The hence defined global variable $t$ represents time passing with constant slope $1$, initialized to $0$.

\subsubsection{Controller.}
\label{subsubsec:ModelingController}
The functional behavior of a controller is provided as a discrete program $\ctrl$
and the associated \emph{reactivity} $\reactivity$. 
The controller acts at least every $\reactivity$ units of time, see Def.~\ref{def:controller}.

\begin{definition}[Reactive Controller]\label{def:controller}
A \emph{reactive controller} $\textbf{RCtrl}(\ctrl,\reactivity)$ with \emph{reactivity boundary} $\reactivity$ and fresh timestamp $\timestamp$ has the program shape 
\[
  \textbf{RCtrl}(\ctrl,\reactivity) \doteq~ (\ptest{t \leq \timestamp + \reactivity};~ \ctrl;~ \pumod{\timestamp}{t}) \enspace ~
 \]
\end{definition}

Execution periodicity is ensured by a \emph{fresh} variable $\timestamp$ time stamping ($\pumod{\timestamp}{t}$) the last execution of $\ctrl$. The prefixing guard $\ptest{t \leq \timestamp + \reactivity}$ forces $\ctrl$ to be executed within $\reactivity$ time since its last execution $\timestamp$. This pattern models control frequency, since all runs not satisfying a test are aborted, see Sec.~\ref{sec:diff-dyn-logic}. 

\begin{example}[Water-level Controller]
 \label{ex:BehaviorWlCtrl}
 We consider the water level controller in a water plant. 
 When the level reaches a maximum (resp. minimum) threshold, we close the inlet faucet $\fin$ (resp. we open the inlet faucet). The resulting controller has the program shape  $\textbf{RCtrl}(\wlctrl,\reactivity_\wlctrl)$ i.e. $(\ptest{t \leq \timestamp + \reactivity_\wlctrl};~ \wlctrl;~ \pumod{\timestamp}{t})$,
where $\reactivity_\wlctrl = 0.05s$ ensures a control frequency of at least 20$\textit{Hz}$. 
The controller:
 \[
 \wlctrl \doteq \pumod{wlm}{wl}; \big(\pchoice{(\ptest{wlm \geq 6.5}; \pumod{\fin}{0})}{(\ptest{wlm \leq 3.5}; \pumod{\fin}{1})} \big)
 \]
 measures the water level using $\pumod{wlm}{wl}$ and then sets $\fin$ depending on whether the water level exceeds the minimum threshold $3.5$ or the maximum threshold $6.5$.
This controller makes implicit assumptions on the maximum inflow and outflow of the water tank through the relation between its reactivity $\reactivity_\wlctrl$ and the thresholds on $wlm$.
\end{example}

\subsubsection{Plant.}
\label{subsubsec:Plant}

The functional behavior of the plant is provided as an ODE system $\pevolvein{\D{x}=\theta}{H}$ with $t \not\in \vars{\D{x}=\theta}$ and the \emph{controllability bound} $\controllability$. Controllability is implemented by 
adding the formula $t \leq \controllability$ to the evolution domain, see Def.~\ref{def:plant}.

\begin{definition}[Controllable Plant]\label{def:plant}
A \emph{controllable plant} $\textbf{CPlant}(\pevolvein{\D{x}=\theta}{H}, \controllability)$ with \emph{controllability bound} $\controllability$ is an ODE system of the shape 
 \(
 \textbf{CPlant}(\pevolvein{\D{x}=\theta}{H},\controllability) \doteq \pevolvein{\D{x}=\theta}{H \et t \leq \controllability}
 \), combined with time defined by $\textbf{Time}(t)$.
\end{definition}

\begin{example}[Water-level]
\label{ex:BehaviorWl}
The evolution of the water level $wl$ in the tank is determined by the difference between the inlet flow $\fin$ and the outlet flow $\fout$.
The water level is always non-negative ($H \doteq wl \geq 0$), and so the controllable water level is the ODE with controllability $\controllability_\wl = 0.2s$: 
 \[
   \pevolvein{\D{wl} = \fin - \fout, \D{t}=1}{t\geq 0\et wl \geq 0 \et t \leq \controllability_\wl}
 \]
\end{example}

We compose the plant with the controller to a full system with repeated interaction between the plant and the controller.

\subsubsection{Full system.}

The full system is obtained by applying parallel composition as defined in~\eqref{eq:choiceparallel} to the plant and the controller, but with one important change: the formula $t \leq \controllability$ is replaced by the formula $t \leq \timestamp + \reactivity$ to ensure that the plant suspends when the controller is expected to run, see Def.~\ref{defi:ProposedEncoding}. 

\begin{definition}[Computer-Controlled System]
 \label{defi:ProposedEncoding}
A \emph{computer-controlled system} $\textbf{CCS}$ is a parallel composition of a reactive controller $\textbf{RCtrl}(\ctrl, \reactivity)$ and a controllable plant $\textbf{CPlant}(\pevolvein{\D{x}=\theta}{H}, \controllability)$ with $\reactivity \leq \controllability$ and the resulting hybrid program shape, assuming $\textbf{Time}(t)$ and, initially, $\timestamp=t$:
\[
  \textbf{CCS} \doteq \prepeat{\big( \pevolvein{\D{x}=\theta}{H \et \underbrace{t \leq \timestamp + \reactivity}_{\reactivity \leq \controllability}}
\cup \textbf{RCtrl}(\ctrl,\reactivity) \big)}
\]
\end{definition}

Execution between the controller and the plant switches based on the variable $\timestamp$. 
At the beginning of each loop iteration, we have $t \geq \timestamp$.
The difference $t - \timestamp$ grows according to the evolution of time until the point $t - \timestamp = \reactivity$. 
At the latest, then, the controller must act before the plant can continue. 
Safety requires $\reactivity \leq \controllability$, \ie the reactivity of the controller is at most the controllability of the plant. 
Otherwise, there may be runs of the whole system where the controller executes too late for the plant to stay safe. 

\begin{example}[Water-tank]
\label{ex:BehaviorWT}
 We compose the water level with its controller to obtain the water tank system with the following behavior:
 \[
  \left(
  \begin{array}{l}
   \ \ \ \pevolvein{\D{wl} = \fin - \fout, \D{t}=1}{t\geq 0\et wl \geq 0 \et t \leq \timestamp + \reactivity_\wlctrl} \\
   \cup  \ (\ptest{t \leq \timestamp + \reactivity_\wlctrl}; \wlctrl; \pumod{\timestamp}{t}) 
  \end{array}
  \right)^\ast
 \]
The composition is possible because the reactivity of the controller ($\reactivity_\wlctrl = 0.05s$) does not exceed the controllability of the plant ($\controllability_\wl = 0.2s$). 

\end{example}



\subsection{Modular Verification of a CCS}
\label{subsec:ModularVerificationCCS}


Based on the modular modeling capabilities offered by the concepts of Sec.~\ref{subsec:ModelingCCS} and through~\cite[Thm.~2]{Lunel2017}, we provide techniques to verify the safety of a complete system from safety proofs of its components (which can be reactive controllers, controllable plants, or subsystems built from those following the computer-controlled systems composition).
The proofs of our theorems are syntactic using the axioms of $\dL$ (as opposed to the semantic proofs in~\cite{Lunel2017}) and are, thus, implementable as tactics in the theorem prover \KeYmaeraX~\cite{DBLP:conf/cade/FultonMQVP2015}.

\subsubsection{Environment.} A description of the global system environment $\Env$ characterizing constants (either as exact values or through their relevant characteristics) is necessary. We require that $\freevars{\Env} \cap \boundvars{\alpha} = \emptyset$ for all system components $\alpha$ to ensure that the environment variables are constants: these constants are not controlled, but can be read by all components. 
(\eg, gravity constant $g$).

\begin{example}[Water tank environment]
In the water tank example, the environment $\Env_{wt} \doteq \fout = 0.75 \et \reactivity_\wlctrl = 0.05 \et \controllability_\wl = 0.2$ is the outlet flow $\fout$ of $0.75$, plant controllability $\controllability_\wl$ of $0.2s$, and controller reactivity $\reactivity_\wlctrl$ of $0.05s$.
\end{example}

\subsubsection{Contracts.}
\label{subsubsec:ControllerProof}

A designer specifies the assumptions $A_\ctrl$ and guarantees $G_\ctrl$ of the controller as well as the assumption $A_\plant$ and guarantees $G_\plant$ of the plant. 
In order to be compositional, the guarantees of the controller must not refer to outputs of the plant and inversely ($\freevars{G_\ctrl} \cap \boundvars{\plant} = \emptyset$ and $\freevars{G_\plant} \cap \boundvars{\ctrl} = \emptyset$).
A component $\alpha$ satisfies its contract $(A_\alpha,G_\alpha)$ in environment $\Env$ under starting conditions $\Init_\alpha$ if formula \(
 (\Env \et A_\alpha \et \Init_\alpha) \imply [\prepeat{\alpha}] G_\alpha
\) is valid (\eg, proved using the $\dL$ proof calculus).
Unlike the environment $\Env$, the initial conditions $\Init_\alpha$ and assumptions $A_\alpha$ of a component $\alpha$ can mention assumptions about the state of other components.

\begin{example}[Water tank contracts]
\label{ex:ContractWlCtrl}
The water-level controller assumes that the actual water level in the tank ranges over the interval $[3,7]$ (as guaranteed by the tank), and itself guarantees to drain the tank when the measured water level approaches the upper threshold $6.5$, and fill the tank when below the lower threshold $3.5$.
The tank contract assumes that the tank is instructed correctly to drain or fill, and then guarantees to keep the water level in the limits $[3,7]$.

\vspace{-.2\baselineskip}
\begin{minipage}{.6\linewidth}
 \[
  \left\{
   \begin{array}{c l}
    A_\wlctrl: & G_\wl \\
    G_\wlctrl: & wlm \leq 3.5 \imply \fin = 1 \\
       \ & 6.5 \leq wlm \imply \fin = 0 \\
       \ & (3.5 \leq wlm \leq 6.5) \imply (\fin = 0 \ou \fin = 1)
 \end{array}
 \right.
\]
\end{minipage}
\begin{minipage}{.35\linewidth}
\[
  \left\{
   \begin{array}{c l}
    A_\wl: & G_\wlctrl \\
    G_\wl: & 3 \leq wl \leq 7
 \end{array}
 \right.
\]
\end{minipage}
\vspace{.2\baselineskip}

These contracts assume that the measured water level is correct, \ie it corresponds to the true water level in the tank, so $\Init_\ctrl \equiv wl=wlm$ and also $\Init_\plant \equiv wl=wlm$.
As we compose the controller and plant components to a full system, where the plant evolves for some time between controller runs (and thus measurements), we will need to find a condition that describes how the true water level evolves in relationship to the measured water level.
\end{example}


\subsubsection{Full system.}
\label{subsubsec:CompositionCCS}

The contract $(A_\ctrl \et A_\plant, G_\ctrl \et G_\plant)$ for the full system is the conjunction of the assumptions and of the guarantees. 

\paragraph{Composition invariant.}
In the full system, the controller and the plant will run in a quasi-parallel fashion, so time passes between controller runs and thus in turn between measurements of the true plant values.
With a composition invariant $J_\com$ we describe the relationship between the true values of the plant and the measured values in the controller. The formula $J_\com$ is a composition invariant for two components $\alpha$ and $\beta$ if the formulas $J_\com \limply \dbox{\alpha}J_\com$ and
$J_\com \limply \dbox{\beta}J_\com$ are valid (components maintain the composition invariant), and $\Init_\alpha \land \Init_\beta \limply J_\com$ is valid (composition invariant is initially satisfied).

Each component is responsible for satisfying its own guarantees and can assume that others will satisfy its assumptions. 
We also require that other components do not interfere with a component's guarantees.
This notion of non-interference ensures that contracts
focus on the behavior of their own component (but nothing else), as intuitively expected.

\begin{definition}[Non-interfering Controller and Plant]
  A controller $\ctrl$ and plant $\{\D{x}=\theta\&H\}$ are \emph{non-interfering} if they do not influence the guarantees of the respective other component, so $\freevars{G_\ctrl} \cap \boundvars{\{\D{x}=\theta\&H\}}=\emptyset$
  and $\freevars{G_\plant} \cap \boundvars{\ctrl}=\emptyset$, and if they do not share the same outputs, so $\boundvars{\ctrl} \cap \boundvars{\{\D{x}=\theta\&H\}}=\emptyset$.
\end{definition}

For composition it is important that contracts are compatible, meaning that they mutually satisfy their assumptions from their respective guarantees.

\begin{definition}[Compatible Contracts]
 Contracts $(A_\alpha, G_\alpha)$ and $(A_\beta, G_\beta)$ of components $\alpha$ and $\beta$ with composition invariant $J_\com$ are \emph{compatible} if the formulas $A_\alpha \limply \dbox{\beta}(G_\beta \land J_\com \limply A_\alpha)$ and $A_\beta \limply \dbox{\alpha}(G_\alpha \land J_\com \limply A_\beta)$ are valid. 
\end{definition}

\begin{theorem}[Safe composition of controller and plant]
  \label{thm:CompoCCS}
  Let $\textbf{RCtrl}(\ctrl,\reactivity)$ be a reactive controller satisfying its contract $(A_\ctrl,G_\ctrl)$ and
  $\textbf{CPlant}(\pevolvein{\D{x}=\theta}{H},\controllability)$ be a controllable plant satisfying its contract $(A_\plant,G_\plant)$.
  Further let the components $\textbf{RCtrl}(\ctrl,\reactivity)$ and $\textbf{CPlant}(\pevolvein{\D{x}=\theta}{H},\controllability)$ be non-interfering, the contracts $(A_\ctrl,G_\ctrl)$ and $(A_\plant,G_\plant)$ be compatible, and $J_\com$ be a composition invariant. 
  Then, the parallel composition $\textbf{CCS}$ is safe, i.e.,
  $(\Env \land A_\ctrl \land \Init_\ctrl \land A_\plant \land \Init_\plant) \limply \dbox{\textbf{CCS}}(G_\ctrl \land G_\plant)$ is valid.
\end{theorem}

\begin{proof}
\label{proof:CompoCCS}
Adapts \cite[Thm. 2]{Lunel2017} to a syntactic $\dL$ proof with loop invariant $A_\ctrl \land G_\ctrl \land A_\plant \land G_\plant \land J_\com$ with differential refinement to replace $\reactivity$ with $\controllability$.
Prove with loop invariant $A_\ctrl \land G_\ctrl \land A_\plant \land G_\plant \land J_\com$:
\begin{description}
\item[Base Case] $\Env \land \Init_\ctrl \land \Init_\plant \land A_\ctrl \land A_\plant \limply G_\ctrl \land G_\plant \land J_\com$: follows from base cases of the component induction proofs $\Env \land \Init_\ctrl \land A_\ctrl \limply A_\ctrl \land G_\ctrl$ and $\Env \land \Init_\plant \land A_\plant \limply A_\plant \land G_\plant$, as well as initially satisfied composition invariant $\Init_\ctrl \land \Init_\plant \limply J_\com$.
\item[Use Case] $A_\ctrl \land G_\ctrl \land A_\plant \land G_\plant \land J_\com \limply G_\ctrl \land G_\plant$: trivial.
\item[Induction Step] The induction step unfolds into the following cases (unfolding the non-deterministic choice and the conjuncts of the loop invariant):
\begin{enumerate}
\item $A_\ctrl \land G_\ctrl \land A_\plant \land G_\plant \land J_\com \limply \dbox{\textbf{RCtrl}(\ctrl,\reactivity)}(A_\ctrl \land G_\ctrl)$: follows from induction step of component proof $G_\ctrl \limply \dbox{\ctrl}G_\ctrl$.
\item $A_\ctrl \land G_\ctrl \land A_\plant \land G_\plant \land J_\com \limply \dbox{\textbf{RCtrl}(\ctrl,\reactivity)}J_\com$: follows from composition invariant proof $J_\com \limply \dbox{\ctrl}J_\com$.
\item $A_\ctrl \land G_\ctrl \land A_\plant \land G_\plant \land J_\com \limply \dbox{\textbf{RCtrl}(\ctrl,\reactivity)}G_\plant$: follows from $\freevars{G_\plant} \cap \boundvars{\ctrl}=\emptyset$.
\item $A_\ctrl \land G_\ctrl \land A_\plant \land G_\plant \land J_\com \limply \dbox{\textbf{RCtrl}(\ctrl,\reactivity)}A_\plant$: follows from Case 1+2 with compatibility.
\item $A_\ctrl \land G_\ctrl \land A_\plant \land G_\plant \land J_\com \limply \dbox{\textbf{CPlant}(\pevolvein{\D{x}=\theta}{H},\reactivity)}G_\ctrl$: follows from $\freevars{G_\ctrl} \cap \boundvars{\pevolvein{\D{x}=\theta}{H}}=\emptyset$.
\item $A_\ctrl \land G_\ctrl \land A_\plant \land G_\plant \land J_\com \limply \dbox{\textbf{CPlant}(\pevolvein{\D{x}=\theta}{H},\reactivity)}(A_\plant \land G_\plant)$: follows from differential refinement ($\reactivity \leq \controllability$) with induction step of component proof $A_\plant \land G_\plant \limply \dbox{\textbf{CPlant}(\pevolvein{\D{x}=\theta}{H},\controllability)}(A_\plant \land G_\plant)$.
\item $A_\ctrl \land G_\ctrl \land A_\plant \land G_\plant \land J_\com \limply \dbox{\textbf{CPlant}(\pevolvein{\D{x}=\theta}{H},\reactivity)}J_\com$: follows from differential refinement ($\reactivity \leq \controllability$) with composition invariant proof $J_\com \limply \dbox{\textbf{CPlant}(\pevolvein{\D{x}=\theta}{H},\controllability)}J_\com$.
\item $A_\ctrl \land G_\ctrl \land A_\plant \land G_\plant \land J_\com \limply \dbox{\textbf{CPlant}(\pevolvein{\D{x}=\theta}{H},\reactivity)}A_\ctrl$: follows from Case 6+7 with compatibility.
\end{enumerate}
\end{description}
\qed
\end{proof}

\begin{example}[Water-tank contract]
The controller and the water-level are non-interfering, their contracts compatible, and the controller is fast enough to keep the plant safe ($\reactivity_\wlctrl \leq \controllability_\wl$). 
We apply Thm.~\ref{thm:CompoCCS} with the composition invariant $J_\com \doteq wl = (\fin - \fout)(t - \timestamp) + wlm$ to obtain that the composition is safe, \ie formula $\Env \land \Init_\wl \land \Init_\wlctrl \land A_\wl \land A_\wlctrl \imply [\textbf{Water-tank}] (G_\wl \land G_\wlctrl)$ is valid.  The composition invariant says how the true value $\wl$ deviates from the last measured value $wlm$ according to the flow $\fin-\fout$ as time $t-\timestamp$ passes. 
\end{example}


\paragraph{Outlook}

We adapted parallel composition of~\cite{Lunel2017} to model and prove computer-controlled systems composed of two components, a reactive controller and a controllable plant. Next, we extend this concept to arbitrarily nested combinations of controllers and plants with a systematic integration of timed constraints.

\section{Parallel Composition}
\label{sec:TimedParallelComposition}

We want to extend the integration of temporal considerations for every component in a timed framework. The previous section shows the importance of temporal considerations in CCS. Industrial systems combine CCS in parallel and it is necessary to have a framework to handle temporal properties.

In order to reason about parallel execution of control software sharing computation resources, models of different costs (controllability, performance, latency, etc) become important.
For example, when two programs execute quasi-parallel on a single CPU core, their computation resources are shared and execution may mutually preempt.
As a result, the worst-case execution times of the programs sum up to the total worst-case execution time of the composed system.
This requires designing plants with sufficiently longer controllability periods, and controllers that react further in advance.

Based on the parallel composition pattern in~\cite{Lunel2017} and the concepts of reactive controller and controllable plant introduced above, here we present parallel compositions of component hierarchies, including composition of multiple reactive controllers, multiple controllable plants, and mixed compositions. 
We retain the algebraic properties of~\cite{Lunel2017}, commutativity and associativity, and present theorems guaranteeing that the conjunction of contracts is preserved through composition.





Controllable plants are already hierarchically compositional per Def.~\ref{def:plant}.
The particular structure to enclose control programs with temporal guards in reactive controllers, however, makes it necessary to extend the definition of reactive controller (Def.~\ref{def:controller}) to a \emph{multi-choice reactive controller} that combines choices of each of its constituting atomic reactive controllers non-deterministically.
We associate a \emph{fresh variable} $\timestamp_i$ with each atomic reactive controller $\ctrl_i$. It is used to specify the time stamp of the controller in an execution cycle.  


\begin{definition}[Multi-Choice Reactive Controller]\label{def:multiChoicesController}
A \emph{multi-choice reactive controller} $\textbf{MRCtrl}\left(\bigcup_{1\leq i \leq n} \ctrl_i,\reactivity\right)$ with $n$ control choices and overall \emph{reactivity bound} $\reactivity$ has the program shape 
\[
  \textbf{MRCtrl}\left(\bigcup_{1\leq i \leq n} \ctrl_i,\reactivity\right) \doteq~ \left(\bigcup_{1 \leq i \leq n}\textbf{RCtrl}(\ctrl_i,\reactivity) \right) \enspace .
 \]
\end{definition}


%

The parallel composition follows cases for purely discrete components, purely continuous components or a mix of both, which we detail in the subsections below. We illustrate each case with an example with two connected water-tanks, one where the inlet flow of one is the outlet flow of the other, with respective reactive controllers to ensure that they remain within a pre-defined range. The first controller actuates on the inlet flow of the first tank, whereas the second actuates on the outlet valve of the second tank.



\subsection{Parallel Composition of Multi-Choice Reactive Controllers}
\label{subsec:CompoDisc}


We refine the parallel composition operator for multi-choice reactive controllers to considering the controllability and reactivity bounds $\Delta$ and $\reactivity$ of its components. 
%
%
By definition, the controllability bound of composed components $\alpha$ and $\beta$ is always $\min(\Delta_\alpha,\Delta_\beta)$ of their individual bounds $\Delta_\alpha$, $\Delta_\beta$. The reactivity bound depends on the physical architecture that composes $\alpha$ and $\beta$.  It is overapproximated by a max+ cost function $\C:\mathbb{R}^2\rightarrow\mathbb{R}$
such that $\C(\delta_\alpha,\delta_\beta)=\max(\delta_\alpha,\delta_\beta)$ if $\alpha$ and $\beta$ have controllers running independently (e.g. two ECUs or PLCs), or else $\delta_\alpha+\delta_\beta$, if both controllers execute on one resource. Notice that such a definition is associative and commutative with respect to composition.

\paragraph{Modeling.} We first define the parallel composition of discrete components, which are multi-choice reactive controllers $\textbf{MRCtrl}(\bigcup_{1 \leq i \leq n}\ctrl_i,\reactivity)$. To the definition in~\cite{Lunel2017}, we add the cost model $\C$ to combine individual bounds $\reactivity$ as that of the composed system.
The parallel composition is the non-deterministic choice between all control choices in multi-choice reactive controllers $\textbf{MRCtrl}(\bigcup_{1 \leq i \leq n_\alpha}\alpha_i,\reactivity_\alpha)$ and $\textbf{MRCtrl}(\bigcup_{1 \leq j \leq n_\beta}\beta_j,\reactivity_\beta)$, but with the individual $\reactivity_\alpha$ and $\reactivity_\beta$ replaced by the cost model $\C(\reactivity_\alpha, \reactivity_\beta)$. 
Interleaving of controller executions occurs through embedding the non-deterministic choice in the loop of a full system, see Thm.~\ref{thm:TimedParallelCompoDisc}.


\begin{definition}[Parallel Composition of Multi-Choice Reactive \\ Controllers]
 \label{defi:timed_parallel_compo_disc}
 Let $\alpha$ and $\beta$ be multi-choice reactive controllers with respective program shapes $\textbf{MRCtrl}(\bigcup_{1 \leq i \leq n_\alpha}\alpha_i,\reactivity_\alpha)$ and $\textbf{MRCtrl}(\bigcup_{1 \leq j \leq n_\beta}\beta_j,\reactivity_\beta)$. The parallel composition $\alpha \para \beta$ has the program shape:
\[
\textbf{MRCtrl}\left(\bigcup_{1 \leq i \leq n_\alpha} \alpha_i \cup \bigcup_{1 \leq j \leq n_\beta} \beta_j, \C(\reactivity_\alpha, \reactivity_\beta) \enspace \right) .
\]
\end{definition}

\begin{example}[Composition of two water-level controllers]
\label{ex:compo_water_level_controllers}
We compose two reactive water-level controllers $\wlctrl_1$ (reactivity $\reactivity_{\wlctrl_1} = 0.05$s) and $\wlctrl_2$ (reactivity $\reactivity_{wlctrl_2} = 0.02$s) on one CPU. The multi-choice reactive controller resulting from cost model $\C(\reactivity_{\wlctrl_1},\reactivity_{\wlctrl_2}) =\reactivity_{\wlctrl_1}+\reactivity_{\wlctrl_2}$ is:
 \[
  \begin{array}{l l}
   &\textbf{MRCtrl}\left(\wlctrl_1 \cup \wlctrl_2, \reactivity_{\wlctrl_1} + \reactivity_{\wlctrl_2} \enspace \right) \\
   = &\textbf{RCtrl}(\wlctrl_1,\reactivity_{\wlctrl_1} + \reactivity_{\wlctrl_2}) \cup \textbf{RCtrl}(\wlctrl_2,\reactivity_{\wlctrl_1} + \reactivity_{\wlctrl_2}) \\
   = & \phantom{\cup}(\ptest{t \leq \timestamp_1 + \reactivity_{\wlctrl_1} + \reactivity_{\wlctrl_2}};~ \wlctrl_1;~ \pumod{\timestamp_1}{t}) \\
   & \cup (\ptest{t \leq \timestamp_2 + \reactivity_{\wlctrl_1} + \reactivity_{\wlctrl_2}};~ \wlctrl_2;~ \pumod{\timestamp_2}{t})
  \end{array}
 \]
 where $\wlctrl_1$ follows Example~\ref{ex:BehaviorWlCtrl} and \\
 $\wlctrl_2 \doteq \pumod{wlm_2}{wl}; \big(\pchoice{(\ptest{wlm_2 \geq 9.7}; \pumod{\fout_2}{1})}{(\ptest{wlm_2 \leq 2.3}; \pumod{\fout_2}{0})}$.
\end{example}

%

\paragraph{Algebraic properties.}
We retain commutativity and associativity of the parallel composition operator defined in~\cite{Lunel2017}. Commutativity implies that we are able to decompose a system and associativity ensures that we can build it step-by-step. The proof relies on the commutativity and associativity of both non-deterministic choice and cost model $\C$. 

\begin{proposition}[Commutativity and Associativity for Parallel Composition of Multi-Choice Reactive Controllers]
 \label{prop:AlgCompoDisc}
 Let $\alpha$, $\beta$ and $\gamma$ be multi-choices reactive controllers with respective program shape $\textbf{MRCtrl}\left(\bigcup_{1 \leq i \leq n_\alpha} \alpha_i, \reactivity_\alpha \enspace \right)$, 
 $\textbf{MRCtrl}\left(\bigcup_{1 \leq j \leq n_\beta} \beta_j, \reactivity_\beta \enspace \right)$ and 
 $\textbf{MRCtrl}\left(\bigcup_{1 \leq k \leq n_\gamma} \gamma_k, \reactivity_\gamma \enspace \right)$. Assume that the max+ cost function $\C$ is commutative and associative. Then:
 \[
  \begin{array}{l c l r}
   \alpha \para \beta &=& \beta \para \alpha &\mbox{ (Commutativity)} \\
   (\alpha \para \beta) \para \gamma &=& \alpha \para (\beta \para \gamma) &\mbox{ (Associativity)}
  \end{array}
 \]
\end{proposition}

\begin{proof}
 \label{proof:AlgCompoDisc}
 We unfold the definitions and use the commutativity (resp. associativity) property of the non-deterministic choice operator and of the max+ cost function$\C$.
 \begin{description}
  \item[Commutativity]
  \[
   \begin{array}{l l r}
    &\alpha \para \beta & \\
    \doteq &\textbf{MRCtrl}\left(\bigcup_{1 \leq i \leq n_\alpha} \alpha_i \cup \bigcup_{1 \leq j \leq n_\beta} \beta_j, \C(\reactivity_\alpha, \reactivity_\beta) \enspace \right) &\mbox{(Unfolding definition)} \\
    \doteq &\textbf{MRCtrl}\left(\bigcup_{1 \leq j \leq n_\beta} \beta_j \cup \bigcup_{1 \leq i \leq n_\alpha} \alpha_i, \C(\reactivity_\alpha, \reactivity_\beta) \enspace \right) &\mbox{(Commutativity of $\cup$ )} \\
    \doteq &\textbf{MRCtrl}\left(\bigcup_{1 \leq j \leq n_\beta} \beta_j \cup \bigcup_{1 \leq i \leq n_\alpha} \alpha_i, \C(\reactivity_\beta, \reactivity_\alpha) \enspace \right) &\mbox{(Commutativity of $S$)} \\
    \doteq &\beta \para \alpha &\mbox{(Folding definition)}
   \end{array}
  \]

  \item[Associativity]
  \[
   \begin{array}{l l}
    &(\alpha \para \beta) \para \gamma \\
    \doteq &\textbf{MRCtrl}\left((\bigcup_{1 \leq i \leq n_\alpha} \alpha_i \cup \bigcup_{1 \leq j \leq n_\beta} \beta_j) \cup \bigcup_{1 \leq k \leq n_\gamma} \gamma_k, \C(\C(\reactivity_\alpha, \reactivity_\beta), \reactivity_\gamma) \enspace \right) \\
    &\mbox{(Unfolding definition)} \\
    \doteq &\textbf{MRCtrl}\left(\bigcup_{1 \leq i \leq n_\alpha} \alpha_i \cup (\bigcup_{1 \leq j \leq n_\beta} \beta_j \cup \bigcup_{1 \leq k \leq n_\gamma} \gamma_k), \C(\C(\reactivity_\alpha, \reactivity_\beta), \reactivity_\gamma) \enspace \right) \\
    &\mbox{(Associativity of $\cup$)} \\
    \doteq &\textbf{MRCtrl}\left((\bigcup_{1 \leq i \leq n_\alpha} \alpha_i \cup \bigcup_{1 \leq j \leq n_\beta} \beta_j) \cup \bigcup_{1 \leq k \leq n_\gamma} \gamma_k, \C(\reactivity_\alpha, \C(\reactivity_\beta, \reactivity_\gamma)) \enspace \right) \\
    &\mbox{(Associativity of $S$)} \\
    \doteq &\alpha \para (\beta \para \gamma) \\
    &\mbox{(Folding definition)}
   \end{array}
  \]  
 \end{description}
 \qed
\end{proof}

\paragraph{Modular verification.}
 We adapt~\cite[Thm. 2]{Lunel2017} by adding the condition that the individual reactivity bound $\reactivity_\alpha$ of a controller $\alpha$ must neither occur in its functional behavior $\bigcup_{1 \leq i \leq n_\alpha}\alpha_i$ nor in its guarantees. Failing to do so may prevent to re-use a component proof. 
%

\begin{definition}[Non-interfering Controllers]
  Two controllers $\alpha$ and $\beta$ are \emph{non-interfering} if they do not modify the same variables, \ie the outputs are separated ($\boundvars{\alpha} \cap \boundvars{\beta} = \emptyset$), and
  if they do not influence the guarantees of the other
  component ($\freevars{G_\alpha} \cap \boundvars{\alpha}=\emptyset$
  and $\freevars{G_\beta} \cap \boundvars{\beta}=\emptyset$).
\end{definition}

\begin{theorem}[Safe composition of Multi-Choice Reactive Controllers]
 \label{thm:TimedParallelCompoDisc}
 Let $\alpha$ and $\beta$ be non-interfering multi-choice reactive controllers with program shape $\textbf{MRCtrl}(\bigcup_{1 \leq i \leq n_\alpha}\alpha_i,\reactivity_\alpha)$ and $\textbf{MRCtrl}(\bigcup_{1 \leq j \leq n_\beta}\beta_j,\reactivity_\beta)$ satisfying their compatible contracts $(A_\alpha,G_\alpha)$ and $(A_\beta,G_\beta)$ and let $J_\com$ be a composition invariant.
%
 Then the parallel composition $\alpha \para \beta$ is safe, i.e., $(\Env \land A_\alpha \land \Init_\alpha \land A_\beta \land \Init_\beta) \limply \dbox{\prepeat{(\alpha \para \beta)}}(G_\alpha \land G_\beta)$ is valid.
\end{theorem}
\begin{proof}
\label{proof:TimedParallelCompoDisc}
Similar to Thm.~\ref{thm:CompoCCS} using the additional condition that $\reactivity_\alpha$ (resp. $\reactivity_\beta$) does not appear in the functional behavior $\bigcup_{1 \leq i \leq n_\alpha}\alpha_i$ (resp. $\bigcup_{1 \leq j \leq n_\beta}\beta_j$) of the controller, nor in its guarantee $G_\alpha$ (resp. $G_\beta$). 
The loop invariant $A_\alpha \land G_\alpha \land A_\beta \land G_\beta \land J_\com$. The base case and use case are the same as for Theorem~\ref{thm:CompoCCS}. The induction step unfolds in the following case:
\begin{enumerate}
\item $A_\alpha \land G_\alpha \land A_\beta \land G_\beta \land J_\com \limply \dbox{\textbf{MRCtrl}(\bigcup_{1 \leq i \leq n_\alpha}\alpha_i,\C(\reactivity_\alpha, \reactivity_\beta))}(A_\alpha \land G_\alpha)$: follows from the condition that $\reactivity_\alpha$ does not occur in the functional behavior $\bigcup_{1 \leq i \leq n_\alpha}\alpha_i$, nor in the guarantee $G_\alpha$, with induction step of component proof $A_\alpha \land G_\alpha \imply \dbox{\textbf{MRCtrl}(\bigcup_{1 \leq i \leq n_\alpha}\alpha_i,\reactivity_\alpha)} (A_\alpha \land G_\alpha)$.
\item $A_\alpha \land G_\alpha \land A_\beta \land G_\beta \land J_\com \limply \dbox{\textbf{MRCtrl}(\bigcup_{1 \leq i \leq n_\alpha}\alpha_i,\C(\reactivity_\alpha, \reactivity_\beta))}J_\com$: follows from the condition that $\reactivity_\alpha$ does not occur in the functional behavior $\bigcup_{1 \leq i \leq n_\alpha}\alpha_i$, nor in the formula $J_\com$, with composition invariant proof $J_\com \limply \dbox{\bigcup_{1 \leq i \leq n_\alpha}\alpha_i}J_\com$.
\item $A_\alpha \land G_\alpha \land A_\beta \land G_\beta \land J_\com \limply \dbox{\textbf{MRCtrl}(\bigcup_{1 \leq i \leq n_\alpha}\alpha_i,\C(\reactivity_\alpha, \reactivity_\beta))}G_\beta$: follows from $\freevars{G_\beta} \cap \boundvars{\bigcup_{1 \leq i \leq n_\alpha}\alpha_i}=\emptyset$.
\item $A_\alpha \land G_\alpha \land A_\beta \land G_\beta \land J_\com \limply \dbox{\textbf{MRCtrl}(\bigcup_{1 \leq i \leq n_\alpha}\alpha_i,\C(\reactivity_\alpha, \reactivity_\beta))}A_\beta$: follows from Case 1+2 with compatibility.

\item $A_\alpha \land G_\alpha \land A_\beta \land G_\beta \land J_\com \limply \dbox{\textbf{MRCtrl}(\bigcup_{1 \leq j \leq n_\beta}\beta_j,\C(\reactivity_\alpha, \reactivity_\beta))}G_\alpha$: follows from $\freevars{G_\alpha} \cap \boundvars{\bigcup_{1 \leq j \leq n_\beta}\beta_j}=\emptyset$.
\item $A_\alpha \land G_\alpha \land A_\beta \land G_\beta \land J_\com \limply \dbox{\textbf{MRCtrl}(\bigcup_{1 \leq j \leq n_\beta}\beta_j,\C(\reactivity_\alpha, \reactivity_\beta))}(A_\beta \land G_\beta)$: follows from the condition that $\reactivity_\beta$ does not occur in the functional behavior $\bigcup_{1 \leq j \leq n_\beta}\beta_j$, nor in the guarantee $G_\beta$ with induction step of component proof $A_\beta \land G_\beta \limply \dbox{\textbf{MRCtrl}(\bigcup_{1 \leq j \leq n_\beta}\beta_j,\reactivity_\beta)}(A_\beta \land G_\beta)$.
\item $A_\alpha \land G_\alpha \land A_\beta \land G_\beta \land J_\com \limply \dbox{\textbf{MRCtrl}(\bigcup_{1 \leq j \leq n_\beta}\beta_j,\C(\reactivity_\alpha, \reactivity_\beta))}J_\com$: follows from the condition that $\reactivity_\beta$ does not occur in the functional behavior $\bigcup_{1 \leq j \leq n_\beta}\beta_j$, nor in the formula $J_\com$, with composition invariant proof $J_\com \limply \dbox{\textbf{MRCtrl}(\bigcup_{1 \leq j \leq n_\beta}\beta_j,\reactivity_\beta)}J_\com$.
\item $A_\alpha \land G_\alpha \land A_\beta \land G_\beta \land J_\com \limply \dbox{\textbf{MRCtrl}(\bigcup_{1 \leq j \leq n_\beta}\beta_j,\C(\reactivity_\alpha, \reactivity_\beta))}A_\alpha$: follows from Case 6+7 with compatibility.
 \end{enumerate}
 \qed
\end{proof}

Non-interference of controllers and compatibility of contracts are standard requirements when modeling a system compositionally and safely. 

%
%
%
%

\begin{example}[Safe composition of two water-level controllers]
 The contract of the first reactive controller $\wlctrl_1$ is the same as in Example~\ref{ex:ContractWlCtrl} with necessary changes. The contract for the second controller is :
 \[
  \left\{
   \begin{array}{c l}
    A_{\wlctrl_2}: & \top \\
    G_{\wlctrl_2}: & wlm_2 \leq 2.3 \imply \fout_2 = 0 \\
       \ & 9.7 \leq wlm_2 \imply \fout_2 = 1 \\
       \ & (2.3 \leq wlm_2 \leq 9.7) \imply (\fout_2 = 0 \ou \fout_2 = 1)
 \end{array}
 \right.
\]
The controller actuates the outlet valve of the system ($\fout_2$). It opens it when the real water-level of the second tank is too close to the maximum threshold ($10$ here) to drain the tank and close it to fill the tank if too close to the minimum threshold ($2$). The two controllers are non-interfering, the contracts are compatible and they both satisfy their contracts (verified using the proof calculus of $\dL$). Hence, Thm.~\ref{thm:TimedParallelCompoDisc} guarantees that the parallel composition is safe, \ie that the contract $(A_{\wlctrl_1} \land A_{\wlctrl_2}, G_{\wlctrl_1} \land G_{\wlctrl_2})$ is valid.
\end{example}

\subsection{Parallel Composition of Controllable Plants}
\label{subsec:CompoCont}


When composing two continuous components in parallel, the controllability of the resulting system is the minimum of their individual controllability bounds (which is obvious from the semantics of ODEs listed in Tab.~\ref{tab:hybrid-programs}: safety proofs hold for any non-negative duration, so also for smaller durations).

\paragraph{Modeling}

Non-interference of controllable plants ensures that their combined continuous dynamics stays true to the isolated dynamics, and that they do not interfere with the guarantees of the respective other component.


%


\begin{definition}[Non-interfering Plants]
Two controllable plants $\alpha$ and $\beta$ with $\textbf{CPlant}(\pevolvein{\D{x}=\theta}{H},\controllability_\alpha)$ and $\textbf{CPlant}(\pevolvein{\D{y}=\eta}{Q},\controllability_\beta)$ and contracts $(A_\alpha, G_\alpha)$ and $(A_\beta, G_\beta)$ are \emph{non-interfering} if $\boundvars{\pevolvein{\D{x}=\theta}{H}} \cap \freevars{\eta}=\emptyset$ and $\boundvars{\pevolvein{\D{y}=\eta}{Q}} \cap \freevars{\theta}=\emptyset$, and if $\boundvars{\pevolvein{\D{x}=\theta}{H}} \cap \freevars{G_\beta} = \emptyset$ and $\boundvars{\pevolvein{\D{y}=\eta}{Q}} \cap \freevars{G_\alpha} = \emptyset$.
\end{definition}
Note that non-interference implies $\boundvars{\pevolvein{\D{x}=\theta}{H}} \cap \boundvars{\pevolvein{\D{y}=\eta}{Q}} = \emptyset$.

\begin{definition}[Parallel Composition of Controllable Plants]
 Let $\alpha$ and $\beta$ be non-interfering controllable plants $\textbf{CPlant}(\pevolvein{\D{x}=\theta}{H},\controllability_\alpha)$ and\linebreak $\textbf{CPlant}(\pevolvein{\D{y}=\eta}{Q},\controllability_\beta)$. The parallel composition $\alpha \para \beta$ is an ODE system of the shape
 \(
   \begin{array}{l}
    \textbf{CPlant}(\pevolvein{\D{x}=\theta, \D{y}=\eta}{H \et Q},min(\controllability_\alpha, \controllability_\beta)) \enspace .
   \end{array}
 \)
\end{definition}

\begin{example}[Composition of two water-level]
 \label{ex:compo_water_level}
Here, we compose the water level dynamics of two tanks ($\pevolvein{\D{wl_1} = \fin - \fout_1, \D{t}=1}{\wl_1 \geq 0 \et t \leq \controllability_{\wl_1}}$ and $\pevolvein{\D{wl_2} = \fout_1 - \fout_2, \D{t}=1}{\wl_2 \geq 0 \et \controllability_{\wl_2}}$) to obtain a controllable plant modeling the evolution of both water levels simultaneously. Their respective controllability bounds are $\controllability_{\wl_1} = 0.2s$ and $\controllability_{\wl_2} = 0.15s$. The controllable plant resulting from the parallel composition expands to

\begin{multline*}
\{\D{wl_1} = \fin - \fout_1, \D{wl_2} = \fout_1 - \fout_2, \D{t}=1\\
\ODEand \wl_1 \geq 0 \et \wl_2 \geq 0 \et t \leq min(\controllability_{\wl_1}, \controllability_{\wl_2}) \}
\end{multline*}
\end{example}

\paragraph{Algebraic properties.}
Commutativity and associativity of the parallel composition pattern defined in~\cite{Lunel2017} are preserved. The proof follows from commutativity and associativity of ``,'' in ODEs and of operator $\min$.

\begin{proposition}[Commutativity and Associativity for Parallel Composition of Controllable Plants]
\label{prop:AlgCompoCont}
 Let $\alpha$, $\beta$ and $\gamma$ non-interfering controllable plants $\textbf{CPlant}(\pevolvein{\D{x}=\theta_x}{H_x},\controllability_\alpha)$, $\textbf{CPlant}(\pevolvein{\D{y}=\theta_y}{H_y},\controllability_\beta)$ and $\textbf{CPlant}(\pevolvein{\D{z}=\theta_z}{H_z},\controllability_\gamma)$. Then:
  \[
  \begin{array}{l c l r}
   \alpha \para \beta &=& \beta \para \alpha &\mbox{ (Commutativity)} \\
   (\alpha \para \beta) \para \gamma &=& \alpha \para (\beta \para \gamma) &\mbox{ (Associativity)}
  \end{array}
 \]
\end{proposition}

\begin{proof}
 \label{proof:AlgCompoCont}
 We unfold the definitions, use the commutativity (resp. associativity) property of the minimum operator $min(\cdot, \cdot)$, of the conjunction $\et$ (for the evolution domain formulas) and of ``,'' in ODEs.
 \begin{description}
  \item[Commutativity]
   \[
    \begin{array}{l l r}
     &\alpha \para \beta & \\
     \doteq &\textbf{CPlant}\big(\pevolvein{\D{x}=\theta_x, \D{y}=\theta_y}{H_x \et H_y},min(\controllability_\alpha, \controllability_\beta)\big) &\mbox{(Unfolding definition)} \\
     \doteq &\textbf{CPlant}\big(\pevolvein{\D{y}=\theta_y, \D{x}=\theta_x}{H_x \et H_y},min(\controllability_\alpha, \controllability_\beta)\big) &\mbox{(Commutativity of ``,'')} \\
     \doteq &\textbf{CPlant}\big(\pevolvein{\D{y}=\theta_y, \D{x}=\theta_x}{H_y \et H_x},min(\controllability_\alpha, \controllability_\beta)\big) &\mbox{(Commutativity of $\et$)} \\
     \doteq &\textbf{CPlant}\big(\pevolvein{\D{y}=\theta_y, \D{x}=\theta_x}{H_y \et H_x},min(\controllability_\beta, \controllability_\alpha)\big) &\mbox{(Commutativity of $min(\cdot, \cdot)$)} \\
     \doteq &\beta \para \alpha &\mbox{(Folding definition)}
    \end{array}
   \]

  \item[Associativity]
  \[
   \begin{array}{l l}
    &(\alpha \para \beta) \para \gamma \\
    \doteq &\textbf{CPlant}\big(\pevolvein{(\D{x}=\theta_x, \D{y}=\theta_y), \D{z}=\theta_z}{(H_x \et H_y) \et H_z},min(min(\controllability_\alpha, \controllability_\beta), \controllability_\gamma) \big) \\ &\mbox{(Unfolding definition)} \\
    \doteq &\textbf{CPlant}\big(\pevolvein{\D{x}=\theta_x, (\D{y}=\theta_y, \D{z}=\theta_z)}{(H_x \et H_y) \et H_z},min(min(\controllability_\alpha, \controllability_\beta), \controllability_\gamma) \big) \\ &\mbox{(Associativity of ``,'')} \\
    \doteq &\textbf{CPlant}\big(\pevolvein{\D{x}=\theta_x, (\D{y}=\theta_y, \D{z}=\theta_z)}{H_x \et(H_y \et H_z)},min(min(\controllability_\alpha, \controllability_\beta), \controllability_\gamma) \big) \\ &\mbox{(Associativity of $\et$)} \\
    \doteq &\textbf{CPlant}\big(\pevolvein{\D{x}=\theta_x, (\D{y}=\theta_y, \D{z}=\theta_z)}{H_x \et(H_y \et H_z)},min(\controllability_\alpha, min(\controllability_\beta, \controllability_\gamma)) \big) \\ &\mbox{(Associativity of $min(\cdot, \cdot)$)} \\
    \doteq &\alpha \para (\beta \para \gamma) \\
    &\mbox{(Folding definition)}
   \end{array}
  \]  
 \end{description}
 \qed
\end{proof}

\paragraph{Modular verification.}
The conjunction of contracts is retained for parallel composition of continuous components, similar to parallel composition of controllers.


\begin{theorem}[Safe Composition of Controllable Plants]
 \label{thm:TimedParallelCompoCont}
 Let $\alpha$ and $\beta$ be two non-interfering controllable plants $\textbf{CPlant}(\pevolvein{\D{x}=\theta}{H},\controllability_\alpha)$ and $\textbf{CPlant}(\pevolvein{\D{y}=\eta}{Q},\controllability_\beta)$ satisfying their respective compatible contracts $(A_\alpha,G_\alpha)$ and $(A_\beta,G_\beta)$, and let $J_\com$ be a composition invariant.
 Then the parallel composition $\alpha \para \beta$ is safe, i.e., $(\Env \land A_\alpha \land \Init_\alpha \land A_\beta \land \Init_\beta) \limply \dbox{\prepeat{(\alpha \para \beta)}}(G_\alpha \land G_\beta)$ is valid.
\end{theorem}

\begin{proof}
 \label{proof:CompoPlant}
 Similar to Theorem~\ref{thm:CompoCCS} after separating the non-interfering plants using the right-to-left direction of the differential ghost axiom (DG) \cite{DBLP:journals/jar/Platzer17}.
 \[
  DG \ \ \dbox{\pevolvein{\D{x} = f(x)}{H_x}} p(x) \leftrightarrow \exists y \dbox{\pevolvein{\D{x} = f(x), \D{y} = a(x)y + b(x)}{H_x}} p(x)
 \]
 $A_\alpha \land G_\alpha \land A_\beta \land G_\beta \land J_\com \limply \dbox{\textbf{CPlant}(\pevolvein{\D{x}=\theta_x, \D{y}=\theta_y}{H_x \et H_y},min(\controllability_\alpha, \controllability_\beta)}(A_\alpha \land G_\alpha)$: follows from component's proof $A_\alpha \land G_\alpha \land A_\beta \land G_\beta \land J_\com \limply \dbox{\textbf{CPlant}(\pevolvein{\D{x}=\theta_x}{H_x \et H_y},min(\controllability_\alpha, \controllability_\beta)}(A_\alpha \land G_\alpha)$ using the the differential ghost axiom right-to-left since the plants are non-interfering. 
 \qed
\end{proof}

\begin{example}[Safe composition of two water-level]
 The contract for the first water level is the same as in Example~\ref{ex:ContractWlCtrl} with necessary changes. We guarantee that the water level of the second tank is within $2$ and $10$, provided that there is a controller which reacts appropriately. Its contract is:
  \[
  \left\{
   \begin{array}{c l}
    A_{\wl_2}: & G_{\wlctrl_2} \\
    G_{\wl_2}: & 2 \leq \wl_2 \leq 10
 \end{array}
 \right.
\]
We apply Thm.~\ref{thm:TimedParallelCompoCont} to guarantee that the controllable plant modeling the evolution of water levels in distinct connected tanks is safe, \ie it satisfies the contract $(A{\wl_1} \land A_{\wl_2}, G_{\wl_1} \land G_{\wl_2})$. 
\end{example}

\subsection{Parallel Composition of Multi-Choice Reactive Controllers and Controllable Plants}
\label{subsec:CompoDiscCont}


We present the composition of a multi-choice reactive controller with a controllable plant that may result from the composition of several atomic controllable plants. We lift the definition of $\textbf{CCS}$ (Sec.~\ref{sec:CCS}) to a general integration of controllability and reactivity.

\paragraph{Modeling}



We define a multi computer-controlled system $\textbf{MCCS}$ as the parallel composition of a multi-choice reactive controller with a controllable plant. 

\begin{definition}[Multi Computer-Controlled System]
\label{defi:timed_parallel_compo_disc_cont}
 A \emph{multi computer-controlled system} is a parallel composition of a multi-choice reactive controller $\textbf{MRCtrl}(\bigcup_{1 \leq i \leq n}\ctrl_i,\reactivity)$ and a controllable plant $\textbf{CPlant}(\pevolvein{\D{y}=\theta}{H},\controllability)$. The parallel composition $\textbf{MCCS}$ has the hybrid program shape:
 \[
   \begin{array}{l}
    \prepeat{\bigg( \pevolvein{\D{y}=\theta, \D{t} = 1}{H \et \underbrace{\bigwedge_{1 \leq i \leq n}t \leq \timestamp_i + \reactivity}_{\reactivity \leq \controllability}}
\cup \textbf{MRCtrl}\left(\bigcup_{1 \leq i \leq n} ctrl_i,\reactivity\right) \bigg)}
   \end{array}
 \]
\end{definition}

The formula $\bigwedge_{1 \leq i \leq n}t \leq \timestamp_i + \reactivity$ is the conjunction of the reactivity bounds of all the $n$ sub-controllers $\ctrl_i$. 



\paragraph{Modular verification}

Cor.~\ref{cor:CompoDiscCont} lifts Thm.~\ref{thm:CompoCCS} (for a single controller and a single plant) to multi computer-controlled systems of possibly many controllers with a controllable plant representing multiple simultaneous evolutions. 


\begin{corollary}[Safe Composition of Multi-Choice Reactive Controller and Controllable Plant]
\label{cor:CompoDiscCont}
Let $\textbf{MRCtrl}(\bigcup_{1 \leq i \leq n}\ctrl_i,\reactivity)$ be a multi-choice reactive controller non-interfering with the controllable plant $\textbf{CPlant}(\pevolvein{\D{y}=\theta_y}{H_y},\controllability)$ satisfying their compatible contracts $(A_\ctrl,G_\ctrl)$ and $(A_\plant,G_\plant)$. 
Further let $J_\com$ be a composition invariant. 
%
Then, $\textbf{MCCS}$ is safe, i.e., $(\Env \land A_\ctrl \land \Init_\ctrl \land A_\plant \land \Init_\plant) \limply \dbox{\textbf{MCCS}}(G_\ctrl \land G_\plant)$ is valid.
\end{corollary}

\begin{proof}
 Similar to the proof of Thm.~\ref{thm:CompoCCS}, but with multi-choice reactive controller instead of a single reactive controller.
 \qed
\end{proof}



\subsection{Parallel Composition of Multi Computer-Controlled Systems}
\label{subsec:CompoGeneral}


When composing multi computer-controlled systems, the combined reactivity of all controllers must not exceed the combined (minimum) controllability bounds of the plants. 
Otherwise, safety cannot be guaranteed, as elaborated next.

\paragraph{Modeling}

The parallel composition of two multi computer-controlled systems is similar to the composition of a multi-choice reactive controller with a controllable plant to obtain a multi computer-controlled system $\textbf{MCCS}$, but with extra care for the combined reactivity bounds obtained from the physical cost model $\C$. 

%

\begin{definition}[Parallel Composition of Multi Computer-Controlled \\ Systems]
\label{defi:timed_parallel_compo_general}
 Let $\alpha$ and $\beta$ be two multi computer-controlled systems with shapes $\alpha \doteq \big(\pevolvein{\D{x}=\theta, \D{t} = 1}{H \et \bigwedge_{1 \leq i \leq n}t \leq \timestamp_i + \reactivity_\alpha} \cup \textbf{MRCtrl}(\bigcup_{1 \leq i \leq n} \alpha_i,\reactivity_\alpha)\big)^\ast$,
$\beta \doteq \big(\pevolvein{\D{y}=\eta, \D{t} = 1}{Q \et \bigwedge_{1 \leq j \leq m}t \leq \timestamp_j + \reactivity_\beta} \cup \textbf{MRCtrl}(\bigcup_{1 \leq j \leq m} \beta_j,\reactivity_\beta)\big)^\ast$. The parallel composition $\alpha \para \beta$ has the hybrid program shape:
 \[
   \left(
   \begin{array}{l}
    \textbf{MRCtrl}\big(\bigcup_{1 \leq i \leq n} \alpha_i,\C(\reactivity_\alpha,\reactivity_\beta)\big) \cup \textbf{MRCtrl}\big(\bigcup_{1 \leq j \leq m} \beta_j,\C(\reactivity_\alpha,\reactivity_\beta)\big) \\
    \cup \pevolvein{\D{x}=\theta, \D{y}=\eta, \D{t} = 1}{H \et Q \\
    \qquad\qquad \et \underbrace{\bigwedge_{1 \leq i \leq n}t \leq \timestamp_i + \C(\reactivity_\alpha,\reactivity_\beta) ~\et \bigwedge_{1 \leq j \leq m}t \leq \timestamp_j + \C(\reactivity_\alpha,\reactivity_\beta)}_{\C(\reactivity_\alpha,\reactivity_\beta) \leq min(\controllability_\alpha,\controllability_\beta)}}
   \end{array}
   \right)^\ast
 \]
\end{definition}


\paragraph{Algebraic properties}

We retain the commutativity and associativity properties (under the condition that the provided max+ cost function $\C$ is commutative and associative), essential for a modular component-based approach.

\begin{proposition}[Commutativity and Associativity for Parallel Composition of Multi Computer-Controlled Systems]
 \label{prop:AlgCompoGen}
 Let $\alpha$, $\beta$ and $\gamma$ be multi computer-controlled systems with respective hybrid program shape
 \[
  \begin{array}{l}
   \bigg(\pevolvein{\D{x}=\theta_x, \D{t} = 1}{H_x \et \bigwedge_{1 \leq i \leq n_\alpha}t \leq \timestamp_i + \reactivity_\alpha} \cup \textbf{MRCtrl}\big(\bigcup_{1 \leq i \leq n_\alpha} \alpha_i,\reactivity_\alpha\big) \bigg)^\ast \\
   \bigg(\pevolvein{\D{y}=\theta_y, \D{t} = 1}{H_y \et \bigwedge_{1 \leq j \leq n_\beta}t \leq \timestamp_j + \reactivity_\beta} \cup \textbf{MRCtrl}\big(\bigcup_{1 \leq j \leq n_\beta} \beta_j,\reactivity_\beta\big) \bigg)^\ast \\
   \bigg(\pevolvein{\D{z}=\theta_z, \D{t} = 1}{H_z \et \bigwedge_{1 \leq k \leq n_\gamma}t \leq \timestamp_k + \reactivity^\gamma} \cup \textbf{MRCtrl}\big(\bigcup_{1 \leq k \leq n_\gamma} \gamma_k,\reactivity^\gamma\big) \bigg)^\ast.
  \end{array}
 \]
%
 Then:
 \[
  \begin{array}{l c l r}
   \alpha \para \beta &=& \beta \para \alpha &\mbox{ (Commutativity)} \\
   (\alpha \para \beta) \para \gamma &=& \alpha \para (\beta \para \gamma) &\mbox{ (Associativity)}
  \end{array}
 \]
\end{proposition}

\newpage

\begin{proof}
\label{proof:AlgCompoGen}
 We unfold definitions and apply Propositions~\ref{prop:AlgCompoDisc} and~\ref{prop:AlgCompoCont} to inner multi-choice reactive controllers and controllable plants. \\
 
  \textbf{Commutativity}
   \[
   \begin{array}{ll}
    &\alpha \para \beta \\
    \doteq \bigg(&\textbf{MRCtrl}\big(\bigcup_{1 \leq i \leq n_\alpha} \alpha_i,\C(\reactivity_\alpha,\reactivity_\beta)\big) \cup \textbf{MRCtrl}\big(\bigcup_{1 \leq j \leq n_\beta} \beta_j,\C(\reactivity_\alpha,\reactivity_\beta)\big) \\
    &\cup \pevolvein{\D{x}=\theta_x, \D{y}=\theta_y, \D{t} = 1}{H_x \et H_y \\
    &\et \bigwedge_{1 \leq i \leq n_\alpha}t \leq \timestamp_i + \C(\reactivity_\alpha,\reactivity_\beta) \et \bigwedge_{1 \leq j \leq n_\beta}t \leq \timestamp_j + \C(\reactivity_\alpha,\reactivity_\beta)} \bigg)^\ast \\ 
    &\mbox{(Unfolding definition)} \\
    \doteq \bigg(&\textbf{MRCtrl}\big(\bigcup_{1 \leq j \leq n_\beta} \beta_j,\C(\reactivity_\beta,\reactivity_\alpha)\big) \cup \textbf{MRCtrl}\big(\bigcup_{1 \leq i \leq n_\alpha} \alpha_i,\C(\reactivity_\beta,\reactivity_\alpha)\big) \\
    &\cup \pevolvein{\D{x}=\theta_x, \D{y}=\theta_y, \D{t} = 1}{H_x \et H_y \\    
    &\et \bigwedge_{1 \leq i \leq n_\alpha}t \leq \timestamp_i + \C(\reactivity_\alpha,\reactivity_\beta) \et \bigwedge_{1 \leq j \leq n_\beta}t \leq \timestamp_j + \C(\reactivity_\alpha,\reactivity_\beta)} \bigg)^\ast \\ 
    &\mbox{(Proposition~\ref{prop:AlgCompoDisc})} \\
    \doteq \bigg(&\textbf{MRCtrl}\big(\bigcup_{1 \leq j \leq n_\beta} \beta_j,\C(\reactivity_\beta,\reactivity_\alpha)\big) \cup \textbf{MRCtrl}\big(\bigcup_{1 \leq i \leq n_\alpha} \alpha_i,\C(\reactivity_\beta,\reactivity_\alpha)\big) \\
    &\cup \pevolvein{\D{y}=\theta_y, \D{x}=\theta_x, \D{t} = 1}{H_y \et H_x \\    
    &\et \bigwedge_{1 \leq j \leq n_\beta}t \leq \timestamp_j + \C(\reactivity_\beta,\reactivity_\alpha) \et \bigwedge_{1 \leq i \leq n_\alpha}t \leq \timestamp_i + \C(\reactivity_\beta,\reactivity_\alpha) } \bigg)^\ast \\
    &\mbox{(Proposition~\ref{prop:AlgCompoCont})} \\
    \doteq &\beta \para \alpha \\
    &\mbox{(Folding definition)}
   \end{array}
   \]

\newpage

  \textbf{Associativity}
   \[
   \begin{array}{ll}
    &(\alpha \para \beta) \para \gamma \\
    \doteq \Bigg( &\bigg(\textbf{MRCtrl}\big(\bigcup_{1 \leq i \leq n_\alpha} \alpha_i,\C(\C(\reactivity_\alpha,\reactivity_\beta), \reactivity_\gamma)\big) \cup \textbf{MRCtrl}\big(\bigcup_{1 \leq j \leq n_\beta} \beta_j,\C(\C(\reactivity_\alpha,\reactivity_\beta), \reactivity_\gamma)\big) \bigg) \\
    &\cup \textbf{MRCtrl}\big(\bigcup_{1 \leq k \leq n_\gamma} \gamma_k,\C(\C(\reactivity_\alpha,\reactivity_\beta), \reactivity_\gamma)\big) \\
    &\cup \pevolvein{\D{x}=\theta_x, \D{y}=\theta_y, \D{z}=\theta_z, \D{t} = 1}{(H_x \et H_y) \et H_z \\
    &\et \big(\bigwedge_{1 \leq i \leq n_\alpha}t \leq \timestamp_i + \C(\C(\reactivity_\alpha,\reactivity_\beta), \reactivity_\gamma) \et \bigwedge_{1 \leq j \leq n_\beta}t \leq \timestamp_j + \C(\C(\reactivity_\alpha,\reactivity_\beta), \reactivity_\gamma)\big) \\
    &\et \bigwedge_{1 \leq k \leq n_\gamma}t \leq \timestamp_k + \C(\C(\reactivity_\alpha,\reactivity_\beta), \reactivity_\gamma)} \Bigg)^\ast \\
    &\mbox{(Unfolding definition)} \\ 
    \doteq \Bigg( &\textbf{MRCtrl}\big(\bigcup_{1 \leq i \leq n_\alpha} \alpha_i,\C(\reactivity_\alpha,\C(\reactivity_\beta, \reactivity_\gamma))\big) \cup \bigg(\textbf{MRCtrl}\big(\bigcup_{1 \leq j \leq n_\beta} \beta_j,\C(\reactivity_\alpha,\C(\reactivity_\beta, \reactivity_\gamma))\big) \\
    &\cup \textbf{MRCtrl}\big(\bigcup_{1 \leq k \leq n_\gamma} \gamma_k,\C(\reactivity_\alpha,\C(\reactivity_\beta, \reactivity_\gamma))\big) \bigg) \\
    &\cup \pevolvein{\D{x}=\theta_x, \D{y}=\theta_y, \D{z}=\theta_z, \D{t} = 1}{(H_x \et H_y) \et H_z \\
    &\et \big(\bigwedge_{1 \leq i \leq n_\alpha}t \leq \timestamp_i + \C(\C(\reactivity_\alpha,\reactivity_\beta), \reactivity_\gamma) \et \bigwedge_{1 \leq j \leq n_\beta}t \leq \timestamp_j + \C(\C(\reactivity_\alpha,\reactivity_\beta), \reactivity_\gamma)\big) \\
    &\et \bigwedge_{1 \leq k \leq n_\gamma}t \leq \timestamp_k + \C(\C(\reactivity_\alpha,\reactivity_\beta), \reactivity_\gamma)} \Bigg)^\ast \\
    &\mbox{(Proposition~\ref{prop:AlgCompoDisc})} \\
    \doteq \Bigg( &\textbf{MRCtrl}\big(\bigcup_{1 \leq i \leq n_\alpha} \alpha_i,\C(\reactivity_\alpha,\C(\reactivity_\beta, \reactivity_\gamma))\big) \cup \bigg(\textbf{MRCtrl}\big(\bigcup_{1 \leq j \leq n_\beta} \beta_j,\C(\reactivity_\alpha,\C(\reactivity_\beta, \reactivity_\gamma))\big) \\
    &\cup \textbf{MRCtrl}\big(\bigcup_{1 \leq k \leq n_\gamma} \gamma_k,\C(\reactivity_\alpha,\C(\reactivity_\beta, \reactivity_\gamma))\big) \bigg) \\
    &\cup \pevolvein{\D{x}=\theta_x, \D{y}=\theta_y, \D{z}=\theta_z, \D{t} = 1}{H_x \et (H_y \et H_z) \\
    &\et \bigwedge_{1 \leq i \leq n_\alpha}t \leq \timestamp_i + \C(\reactivity_\alpha,\C(\reactivity_\beta, \reactivity_\gamma)) \et \big(\bigwedge_{1 \leq j \leq n_\beta}t \leq \timestamp_j + \C(\reactivity_\alpha,\C(\reactivity_\beta, \reactivity_\gamma)) \\
    &\et \bigwedge_{1 \leq k \leq n_\gamma}t \leq \timestamp_k + \C(\reactivity_\alpha,\C(\reactivity_\beta, \reactivity_\gamma)) \big)} \Bigg)^\ast \\
    &\mbox{(Proposition~\ref{prop:AlgCompoCont})} \\
    \doteq &\alpha \para (\beta \para \gamma) \\
    &\mbox{(Folding definition)}
   \end{array}
   \]
\qed
\end{proof}

\paragraph{Modular verification}

We retain also the respective contracts through the parallel composition.
We assume that the individual reactivity bound $\reactivity_\alpha$ of the controller must not occur in its functional behavior, nor in its guarantees.


\begin{theorem}[Safe composition of Multi Computer-Controlled Systems]
\label{thm:CompoMCCS}
Let $\alpha$ and $\beta$ be non-interfering multi computer-controlled systems (with program shape $\textbf{MCCS}_\alpha$ and $\textbf{MCCS}_\beta$) satisfying their respective compatible contracts $(A_\alpha,G_\alpha)$ and $(A_\beta,G_\beta)$, and let $J_\com$ be a composition invariant.
%
Then the parallel composition $\alpha \para \beta$ is safe, i.e., $(\Env \land A_\alpha \land \Init_\alpha \land A_\beta \land \Init_\beta) \limply \dbox{\alpha \para \beta}(G_\alpha \land G_\beta)$ is valid.
\end{theorem}
 
\begin{proof}
 \label{proof:MCCS}
 The parallel composition of two multi computer-controlled systems $\alpha$ and $\beta$  
 with hybrid program shape $\textbf{MCCS}_\alpha \doteq \pevolvein{\D{x}=\theta_x, \D{t} = 1}{H_x \et \bigwedge_{1 \leq i \leq n}t \leq \timestamp_i + \reactivity_\alpha} \cup \textbf{MRCtrl}\big(\bigcup_{1 \leq i \leq n} \alpha_i,\reactivity_\alpha\big)$ 
 and $\textbf{MCCS}_\beta \doteq \pevolvein{\D{y}=\theta_y, \D{t} = 1}{H_y \et \bigwedge_{1 \leq j \leq m}t \leq \timestamp_j + \reactivity_\beta} \cup \textbf{MRCtrl}\big(\bigcup_{1 \leq j \leq m} \beta_j,\reactivity_\beta\big)$ has the hybrid program shape:
 \[
   \left(
   \begin{array}{l}
    \textbf{MRCtrl}\big(\bigcup_{1 \leq i \leq n} \alpha_i,\C(\reactivity_\alpha,\reactivity_\beta)\big) \cup \textbf{MRCtrl}\big(\bigcup_{1 \leq j \leq m} \beta_j,\C(\reactivity_\alpha,\reactivity_\beta)\big) \\
    \cup \pevolvein{\D{x}=\theta_x, \D{y}=\theta_y, \D{t} = 1}{H_x \et H_y \\
    \et \bigwedge_{1 \leq i \leq n}t \leq \timestamp_i + \C(\reactivity_\alpha,\reactivity_\beta) \et \bigwedge_{1 \leq j \leq m}t \leq \timestamp_j + \C(\reactivity_\alpha,\reactivity_\beta)}
   \end{array}
   \right)^\ast
 \]
 
 The parallel composition of respective inner multi-choice reactive controllers $\textbf{MRCtrl}\big(\bigcup_{1 \leq i \leq n} \alpha_i,\C(\reactivity_\alpha,\reactivity_\beta)\big)$ and $\textbf{MRCtrl}\big(\bigcup_{1 \leq j \leq m} \beta_j,\C(\reactivity_\alpha,\reactivity_\beta)\big)$ results in the multi-choice reactive controller:
 \[
  \textbf{MRCtrl}\big(\bigcup_{1 \leq i \leq n} \alpha_i,\C(\reactivity_\alpha,\reactivity_\beta)\big) \cup  \textbf{MRCtrl}\big(\bigcup_{1 \leq j \leq m} \beta_j,\C(\reactivity_\alpha,\reactivity_\beta)\big)
 \]
 
 The parallel composition of inner controllable plants results in the controllable plant:
 \[
  \textbf{CPlant}\big(\pevolvein{\D{x}=\theta_x, \D{y}=\theta_y}{H_x \et H_y},min(\controllability_\alpha, \controllability_\beta)\big)
 \]
 
 Composing the multi-choice reactive controller with the controllable plant following Def.~\ref{defi:timed_parallel_compo_disc_cont} yields:
 \[
   \left(
   \begin{array}{l}
    \textbf{MRCtrl}\big(\bigcup_{1 \leq i \leq n} \alpha_i,\C(\reactivity_\alpha,\reactivity_\beta)\big) \cup \textbf{MRCtrl}\big(\bigcup_{1 \leq j \leq m} \beta_j,\C(\reactivity_\alpha,\reactivity_\beta)\big) \\
    \cup \pevolvein{\D{x}=\theta_x, \D{y}=\theta_y, \D{t} = 1}{H_x \et H_y \\
    \et \bigwedge_{1 \leq i \leq n}t \leq \timestamp_i + \C(\reactivity_\alpha,\reactivity_\beta) \et \bigwedge_{1 \leq j \leq m}t \leq \timestamp_j + \C(\reactivity_\alpha,\reactivity_\beta)}
   \end{array}
   \right)^\ast
 \]

 We successively apply Theorem~\ref{thm:TimedParallelCompoDisc} to the group of inner controllers, Theorem~\ref{thm:TimedParallelCompoCont} to the global controllable plant and Cor.~\ref{cor:CompoDiscCont} to the composition of both to conclude.
 \qed
\end{proof}

%
%

\paragraph{Outlook}

In this section, we have presented how to extend our previous component-based approach to take into account the timing constraints inherent to the design of a Computer-Controlled System. We have proved that we retain the commutativity and associativity, essential to scale up to realistic systems. Finally, we state and prove theorems to retain contracts through the parallel composition. Theses results give us confidence in the ability of our approach to be adapted to new challenges that will arise when applied to realistic industrial systems.

\section{Related Work}
\label{sec:RelatedWork}

Recent component-based verification techniques~\cite{Muller2016,DBLP:journals/sttt/MullerMRSP18} proposed a composition operator in $\dL$ based on the modeling pattern $\prepeat{(\ctrl;\plant)}$ to split verification of systems into more manageable pieces. It focus on separating self-contained components (a controller monitoring its own plant) instead of separating discrete and continuous fragments.
Morevoer, our approach integrates a max+ cost function which handles timing relations upon composition.

Hybrid automata~\cite{Alur1993} are a popular formalism to model hybrid systems, but composition of automata results in an exponential product automaton which is intractable to analyze in practice. \emph{I/O hybrid automata}~\cite{DBLP:journals/iandc/LynchSV03} is an extension of hybrid automata with explicit inputs and outputs. 
Assume-guarantee reasoning~\cite{Henzinger2001} on such automata tackles composability to prevent state-space explosion. 
Yet, use is in practice restricted to linear hybrid automata. Differential Dynamic Logic is practically used for systems with ODEs (and not just linear ODEs), thus our approach is more expressive.

Hybrid Communicating Sequential Processes (HCSP)~\cite{Jifeng1994} is a hybrid extension of the CSP framework. It features a native parallel composition operator and communicating primitives in addition to standard constructs for hybrid systems (sequences, loops, ODEs) and a proof calculus has been proposed in~\cite{Liu2010}. In contrast, our parallel composition operator is not native and relies on usual constructs of $\dL$. The benefit is that we do not have to extend $\dL$ and check the soundness of such extension, but it requires additional effort to mechanize it into the theorem prover \KeYmaeraX. Also, our approach provide engineering support for timing aspects and modular verification principles.

This paper is an extension of our previous work, with specific features to handle temporal properties of the parallel composition of CCS and syntactic proofs that facilitate implementation of the proposed techniques as tactics in the theorem prover \KeYmaeraX.

\section{Conclusion}
\label{sec:Conclusion}




We presented a component-based verification technique to modularly design and verify computer-controlled systems with special focus on timing constraints (reactivity and controllability) and modular verification. 
Our concepts enable systematic modeling of CCS in a modular way while maintaining algebraic properties of composition operators and preserving contract proofs through composition. 
We additionally support reasoning on non-functional properties (reactivity, controllability) through multiple compositions of reactive controllers and plants. 
This paves the way to ultimately model complex cyber-physical systems (several controllers running in parallel according to a generic max+ cost function that monitor different plants) from only simple, atomic components. 
Verification of safety properties for the global system reduces to component safety proofs with only mild assumptions on the reactivity of controllers (does not exceed the controllability of plants) and compatibility between contracts.

As future work, we intend to allow more aggressive compositions to lift restrictions of the techniques presented here: allow some interference in the parallel composition of controllable plants and reactive controllers with additional compatibility proofs (lift non-interference restriction of Cor.~\ref{cor:CompoDiscCont}); allow time and reactivity in the predictions and guarantees of controllers with refactoring techniques to strengthen control choices upon composition (lift restriction of Thm.~\ref{thm:CompoMCCS} that controllers are not allowed to exploit their reactivity bounds $\reactivity$ for control decisions); and support fine-grained communication going beyond shared variables with communication channels as in Hybrid Communicating Sequential Processes.
For proof automation, we intend to implement the theorems of this paper as tactics in the \KeYmaeraX\ theorem prover.

\bibliographystyle{plain}
\bibliography{biblio}

\begin{thebibliography}{10}

\bibitem{Alur1993}
Rajeev Alur, Costas Courcoubetis, Thomas~A Henzinger, and Pei-Hsin Ho.
\newblock Hybrid automata: An algorithmic approach to the specification and
  verification of hybrid systems.
\newblock In {\em Hybrid systems}, pages 209--229. Springer, 1993.

\bibitem{Benveniste2012}
Albert Benveniste, Beno{\^\i}t Caillaud, Dejan Nickovic, Roberto Passerone,
  Jean-Baptiste Raclet, Philipp Reinkemeier, Alberto Sangiovanni-Vincentelli,
  Werner Damm, Thomas Henzinger, and Kim~Guldstrand Larsen.
\newblock Contracts for system design.
\newblock Technical report, 2012.

\bibitem{DBLP:conf/cade/FultonMQVP2015}
Nathan Fulton, Stefan Mitsch, Jan{-}David Quesel, Marcus V{\"{o}}lp, and
  Andr{\'{e}} Platzer.
\newblock Keymaera {X:} an axiomatic tactical theorem prover for hybrid
  systems.
\newblock In {\em Automated Deduction - {CADE-25} - 25th International
  Conference on Automated Deduction, Berlin, Germany, August 1-7, 2015,
  Proceedings}, pages 527--538, 2015.

\bibitem{Henzinger2001}
Thomas~A Henzinger, Marius Minea, and Vinayak Prabhu.
\newblock Assume-guarantee reasoning for hierarchical hybrid systems.
\newblock In {\em International Workshop on Hybrid Systems: Computation and
  Control}, pages 275--290. Springer, 2001.

\bibitem{Jifeng1994}
He~Jifeng.
\newblock From {CSP} to hybrid systems.
\newblock In {\em A classical mind}, pages 171--189. Prentice Hall
  International (UK) Ltd., 1994.

\bibitem{Liu2010}
Jiang Liu, Jidong Lv, Zhao Quan, Naijun Zhan, Hengjun Zhao, Chaochen Zhou, and
  Liang Zou.
\newblock A calculus for hybrid {CSP}.
\newblock In {\em Asian Symposium on Programming Languages and Systems}, pages
  1--15. Springer, 2010.

\bibitem{Lunel2017}
Simon Lunel, Benoît Boyer, and Jean-Pierre Talpin.
\newblock Compositional proofs in differential dynamic logic.
\newblock In Axel Legay and Klaus Schneider, editors, {\em ACSD}, 2017.

\bibitem{DBLP:journals/iandc/LynchSV03}
Nancy~A. Lynch, Roberto Segala, and Frits~W. Vaandrager.
\newblock Hybrid {I/O} automata.
\newblock {\em Inf. Comput.}, 185(1):105--157, 2003.

\bibitem{Muller2016}
Andreas M{\"u}ller, Stefan Mitsch, Werner Retschitzegger, Wieland Schwinger,
  and Andr{\'e} Platzer.
\newblock A component-based approach to hybrid systems safety verification.
\newblock In Erika Abraham and Marieke Huisman, editors, {\em IFM}, volume 9681
  of {\em LNCS}, pages 441--456. Springer, 2016.

\bibitem{DBLP:journals/sttt/MullerMRSP18}
Andreas M{\"{u}}ller, Stefan Mitsch, Werner Retschitzegger, Wieland Schwinger,
  and Andr{\'{e}} Platzer.
\newblock Tactical contract composition for hybrid system component
  verification.
\newblock {\em STTT}, 20(6):615--643, 2018.
\newblock Special issue for selected papers from FASE'17.

\bibitem{DBLP:conf/lics/Platzer12b}
Andr{\'e} Platzer.
\newblock The complete proof theory of hybrid systems.
\newblock In {\em LICS}, pages 541--550. IEEE, 2012.

\bibitem{DBLP:journals/jar/Platzer17}
Andr{\'e} Platzer.
\newblock A complete uniform substitution calculus for differential dynamic
  logic.
\newblock {\em J. Autom. Reas.}, 59(2):219--265, 2017.

\bibitem{Platzer18}
Andr{\'e} Platzer.
\newblock {\em Logical Foundations of Cyber-Physical Systems}.
\newblock Springer, Cham, 2018.

\bibitem{DBLP:conf/lics/PlatzerT18}
Andr{\'{e}} Platzer and Yong~Kiam Tan.
\newblock Differential equation axiomatization: The impressive power of
  differential ghosts.
\newblock In Anuj Dawar and Erich Gr{\"{a}}del, editors, {\em LICS}, pages
  819--828, New York, 2018. ACM.

\bibitem{frama-C}
Julien Signoles, Pascal Cuoq, Florent Kirchner, Nikolai Kosmatov, Virgile
  Prevosto, and Boris Yakobowski.
\newblock Frama-c: a software analysis perspective.
\newblock volume~27, 10 2012.

\end{thebibliography}

\end{document}